\documentclass[sigconf,nonacm]{acmart}
\AtBeginDocument{%
  \providecommand\BibTeX{{%
    Bib\TeX}}}

\usepackage{amsmath,amssymb,amsfonts}
\usepackage{algorithmic}
\usepackage{graphicx}
\usepackage{textcomp}
\usepackage{xcolor}
\usepackage[T1]{fontenc} 
\usepackage[utf8]{inputenc} 

\def\BibTeX{{\rm B\kern-.05em{\sc i\kern-.025em b}\kern-.08em

    T\kern-.1667em\lower.7ex\hbox{E}\kern-.125emX}}
\newtheorem{theorem}{Theorem}
\newtheorem{proposition}{Proposition}

\newtheorem{definition}{Definition}

\begin{document}

\title{
Differential Privacy Preserving Distributed Quantum Computing
}

\author{Hui Zhong}
\email{hzhong5@uh.edu}
\affiliation{
  \institution{Department of Electrical and Computer Engineering, University of Houston}
  \city{Houston}
  \country{USA}
}

\author{Keyi Ju}
\email{juju@bupt.edu.cn}
\affiliation{
  \institution{State Key Laboratory of Networking and Switching Technology, Beijing University of Posts and Telecommunications}
  \city{Beijing}
  \country{China}
}

\author{Jiachen Shen}
\email{jshen28@CougarNet.UH.EDU}
\affiliation{
  \institution{Department of Electrical and Computer Engineering, University of Houston}
  \city{Houston}
  \country{USA}
}

\author{Xinyue Zhang}
\email{xzhang48@kennesaw.edu}
\affiliation{
  \institution{Department of Computer Science, Kennesaw State University}
  \city{Marietta}
  \country{USA}
}

\author{Xiaoqi Qin}
\email{xiaoqiqin@bupt.edu.cn}
\affiliation{
  \institution{State Key Laboratory of Networking and Switching Technology, Beijing University of Posts and Telecommunications}
  \city{Beijing}
  \country{China}
}

\author{Tomoaki Ohtsuki}
\email{ohtsuki@ics.keio.ac.jp}
\affiliation{
  \institution{Department of Information and Computer Science, Keio University}
  \city{Tokyo}
  \country{Japan}
}

\author{Miao Pan}
\email{mpan2@uh.edu}
\affiliation{
  \institution{Department of Electrical and Computer Engineering, University of Houston}
  \city{Houston}
  \country{USA}
}

\author{Zhu Han}
\email{zhan2@Central.UH.EDU}
\affiliation{
  \institution{Department of Electrical and Computer Engineering, University of Houston}
  \city{Houston}
  \country{USA}
}


\begin{abstract}

Existing quantum computers can only operate with hundreds of qubits in the Noisy Intermediate-Scale Quantum (NISQ) state, while quantum distributed computing (QDC) is regarded as a reliable way to address this limitation, allowing quantum computers to achieve their full computational potential. However, similar to classical distributed computing, QDC also faces the problem of privacy leakage. Existing research has introduced quantum differential privacy (QDP) for privacy protection in central quantum computing, but there is no dedicated privacy protection mechanisms for QDC. To fill this research gap, our paper introduces a novel concept called quantum Rényi differential privacy (QRDP), which incorporates the advantages of classical Rényi DP and is applicable in the QDC domain. Based on the new quantum Rényi divergence, QRDP provides delicate and flexible privacy protection by introducing parameter $\alpha$. In particular, the QRDP composition is well suited for QDC, since it allows for more precise control of the total privacy budget in scenarios requiring multiple quantum operations. We analyze a variety of noise mechanisms that can implement QRDP, and derive the lowest privacy budget provided by these mechanisms. Finally, we investigate the impact of different quantum parameters on QRDP. Through our simulations, we also find that adding noise will make the data less usable, but increase the level of privacy protection.
\end{abstract}




\maketitle

\section{Introduction}
Quantum computing is an emerging field that can process problems in parallel due to quantum properties. The acceleration capability of quantum computers makes them suitable for a variety of fields, such as cryptography~\cite{aumasson2017impact}, molecular simulation~\cite{grimsley2019adaptive}, and combinatorial optimization~\cite{han2000genetic}. In order to fully realize the capabilities of a quantum computer, thousands or even tens of thousands of qubits are usually required to be connected. However, since existing quantum computers are still in the Noisy Intermediate-Scale Quantum (NISQ) era, the available qubits are in the order of hundreds, which limits the possible applications of quantum computers~\cite{tannu2019not}. In this context, it is widely accepted in academia and industry to use distributed quantum computing (DQC) to enable larger-scale quantum computing, i.e., to increase the number of available qubits by bringing together modular and small quantum chips through a quantum network infrastructure~\cite{caleffi2018quantum, denchev2008distributed, barral2024review}. State-of-the-art works have shown that DQC enables multiple quantum processors to work together to solve problems at exponential speeds beyond the processing power of a single quantum computer~\cite{elkin2014can,caleffi2024distributed,barral2024review}.

However, similar to classical distributed systems, DQC faces several data-related issues. While quantum key distribution (QKD) has provided effective means to address the transmission channel's security between QPUs (a specialized processor for performing quantum computations)~\cite{sharma2019survey}, DQC faces the same data privacy issues as quantum computing. Quantum computing is exposed to internal or external attacks such as higher-energy state attacks and quantum-side-channel attacks~\cite{guan2024optimal}, leading to potential privacy breaches. In addition, quantum computing systems are vulnerable to membership inference attacks~\cite{watkins2023quantum} and adversarial attacks~\cite{du2022quantum, aaronson2019gentle}. Therefore, ensuring the protection of data privacy in DQC is essential and critical. In the sequel, we now begin by introducing privacy-preserving methods for classical central and distributed computing. Following that, we discuss existing privacy-preserving techniques for quantum centralized computing and highlight the current gaps in privacy protection methods for QDC.

Differential privacy (DP) has emerged as a promising approach to protect data privacy in traditional centralized computing~\cite{dwork2008differential,qiao2019effective,dankar2013practicing}. DP protects data privacy by adding random noise to the model so that the inclusion or exclusion of individual data does not significantly affect the query result~\cite{dwork2006differential}. Compared to other approaches, DP provides end-to-end data protection and offers a rigorous mathematical formulation, which is often used to measure the degree of privacy protection in terms of $\epsilon$. However, traditional DP still faces challenges when applied to distributed computing environments~\cite{mironov2017renyi}. In distributed computing, the combination theorem of DP sums the privacy loss of multiple operations to obtain the total privacy loss. However, as the number of operations increases, this linear superposition rapidly increases privacy loss, resulting in a combinatorial explosion. Therefore, researchers proposed Rényi differential privacy (RDP) to provide a more flexible privacy computation method that is more suitable for distributed computing scenarios~\cite{mironov2017renyi}. RDP extends the classical DP using Rényi divergence, which introduces Rényi parameter $\alpha$ for finer control of the privacy budget. Rényi divergence is a parametric measure of divergence between probability distributions. RDP can also control the total privacy loss through $\alpha$ to avoid explosion problems. These properties make RDP a compelling choice for enhancing privacy protection in distributed computing scenarios~\cite{mironov2017renyi}.

Given the strengths of DP, it has been further extended to quantum differential privacy (QDP), which has become a promising method to protect data privacy in quantum computing~\cite{zhao2024bridging,ju2024quantum,li2024differential,ju2024privacy,ju2024harnessing,zhong2023tuning}. Zhou et al. in~\cite{zhou2017differential} first introduced the concept of QDP and derived the upper bounds of privacy budgets in the case of three quantum noise mechanisms. As with classical DP, they prove that QDP also satisfies a series of composition theorems. Guan et al. in their CCS publication~\cite{guan2023detecting} introduced the first calculation of the lower bound on the quantum privacy loss. This paper mentions that by focusing on the maximum and minimum eigenvalues of a pre-prepared matrix, the complex problem of verifying whether a quantum algorithm satisfies QDP can be reduced to a computational problem without the need to keep track of the evolution of quantum states.

Although there are several pioneering research efforts for QDP in quantum centralized computing, the privacy preservation study for QDC is still infancy. When it comes to QDC, applying QDP to QDC has some problems as in classical distributed computing, such as privacy loss explosion. Therefore, we would like to employ RDP to QDC to mitigate these issues. However, in the literature, there are no privacy-preserving approaches similar to classical RDP for QDC due to multiple research challenges. For example, how to redefine the appropriate Rényi divergence as well as the privacy budget to satisfy quantum properties. Besides, in order to achieve RDP in the quantum environment, multiple quantum noise mechanisms need to be utilized, which involves investigating the impacts of different quantum noises on the privacy budget. Another challenge is how to accurately calculate and measure the total privacy budget after multiple quantum operations such as DQC. That calls for the novel design of RDP algorithms that are tailored for the quantum environment.

In this paper, we propose quantum Rényi differential privacy (QRDP), which is dedicated to privacy preservation in QDC. To the best of our knowledge, this is the first paper to extend classical RDP to the quantum computing domain. This new privacy preserving framework not only takes advantage of the classical Rényi divergence, but also adapts these principles to the unique characteristics of quantum computing, thus providing good and robust privacy guarantees in the quantum computing setting. Our salient contributions are summarized as follows:
\begin{itemize}
    \item We introduce Quantum Rényi Differential Privacy (QRDP), establishing a novel privacy protection framework tailored for quantum distributed computing.
    
    \item We analyze various quantum noise mechanisms for implementing QRDP, including generalized amplitude damping, composition of phase and amplitude damping, and depolarizing mechanisms. We further propose how to accurately measure the level of privacy protection offered by these mechanisms.
    
    \item We conduct numerical simulations to demonstrate the relationship between QRDP and quantum noise parameters, highlighting the enhanced privacy guarantees offered by QRDP in quantum computing systems.

\end{itemize}

The rest of this paper is organised as follows. In Section~\ref{sec:preliminary}, we introduce quantum computing, quantum distributed learning, existing concepts of QDP. In Section~\ref{sec:QRDP}, we present QRDP and its combination theorem. In Section~\ref{sec:mechanism}, we analyze different noise mechanisms for implementing QRDP and derive formulas for measuring specific privacy budgets. In Section~\ref{sec:QRDPqdp}, we present a QRDP to QDP transformation that is easy to manage the effectiveness of the privacy budget. Finally, we discuss future work in Section~\ref{sec:discussion} and give conclusions in Section~\ref{sec:conclusion}.

\section{Preliminaries\label{sec:preliminary}}
In this section, first, we review the basics of quantum computing, including quantum states, quantum gate operations, and quantum measurements. Next, we introduce both classical DP and QDP for protecting data privacy in centralized computers. Then, we discuss the basic architecture of QDC, which addresses the limitation of quantum centralized computers. Finally, we explore existing RDP frameworks in classical distributed computers, which inspire further consideration of data privacy methods in QDC.


\subsection{Quantum computing}
In digital quantum computing, quantum circuits implement quantum algorithms, including quantum states, quantum gate operations and quantum measurements~\cite{guan2023detecting}. Firstly, we introduce the concept of quantum state, which includes pure and mixed quantum states. The basic unit used in quantum computing is the qubit, which includes the standard orthogonal basis in Hilbert space $\left | 0  \right \rangle$  and $\left | 1  \right \rangle$, which correspond to bits 0 and 1 in classical computation, respectively. In a two-dimensional Hilbert space, a pure quantum state can represent two single qubits $\left | 0  \right \rangle = (1,0)^{T}\ $ and $\left | 1  \right \rangle = (0,1)^{T}$ simultaneously, i.e., $\left | \psi  \right \rangle = \alpha_{1}  \left | 0  \right \rangle + \alpha_{2} \left | 1 \right \rangle = (\alpha_{1},\alpha_{2} )^{T}\in \mathbb{C} ^{2} $, where the complex numbers $\alpha_{1}$ and $\alpha_{2}$ satisfy the formula $\left | \alpha_{1}  \right | ^{2} + \left | \alpha_{2} \right | ^{2} = 1$. In a $2^{n}$-dimensional Hilbert space ($n$ is the number of qubits), pure quantum states can also be represented as similar column vectors, i.e., $\left | \psi  \right \rangle = \sum_{i=1}^{2^{n}}\alpha _{i}\left | i  \right \rangle  \in \mathbb{C} ^{2^{n}}$, which satisfies $\left | \alpha_{1}  \right | ^{2} +...+ \left | \alpha_{2^{n}} \right | ^{2} = 1$. In practice, quantum systems interact with their external environment, so a pure state may collapse into a mixed state, which can be represented by a $2^{n}\times 2^{n}$ positive semi-definite matrix, $\rho =\sum_{i=1}^{2^{n}} p_{i} \left | \psi  \right \rangle \left\langle\psi\right| $. The meaning is that the state is in $\left | \psi  \right \rangle$ with probability $p_{i}$. $\rho$ also satisfies $\text{Tr}\left ( \rho  \right ) =1$. $\text{Tr}(\cdot)$ of $\rho$ is defined as the sum of the diagonal elements of the $\rho$ matrix. In addition, we usually use $\tau(\rho,\sigma)$ to denote the difference between two quantum states $\rho$ and $\sigma$. $\tau$ is also called the trace distance and can be calculated as $\text{Tr}\left ( \left | \rho -\sigma\right |  \right ) /2$.

A quantum circuit consists of the product of a number of sequences of different types of quantum gates $U$. Each $U$ sequence can be mathematically modeled as a matrix. For example, $X$, $Y$, and $Z$ gates are single qubit gates that can be represented as a $2\times 2$ unitary matrix, and $CNOT$ is a two qubit gate that can be represented as a $4\times 4$ unitary matrix. For the input quantum state $\rho$, after the quantum gate operation $U$, its quantum state will become $\rho'=U\rho U^{\dagger}$, which satisfies $U^{\dagger} U = U U^{\dagger} = I$. $U^{\dagger}$ is the adjoint of $U$ and $I$ is the identity operator. 

After going through the quantum circuit, we need to use quantum measurements to get the final calculation results. These results are probability distributions of possible outcomes of quantum measurements. A quantum measurement can be modeled as a set of positive semi-definite matrices $\left \{ M_{k} \right \} _{k\in O}$ on a Hilbert space, and satisfies $ {\textstyle \sum_{k}} M_{k}=I$, where $O$ is the set of measurement results. If we make a measurement of a quantum state $\rho$, there is a probability of $p_{k}$ to get the result $k$, denoted as $p_{k} =\text{Tr}\left \{ M_{k} \rho \right \} $. Such a measurement can also be referred to as Positive Operator-Valued Measure (POVM).

\subsection{Classical and quantum differential privacy for centralized computers}
Classical DP is a mature framework for protecting data privacy in centralized computers~\cite{dwork2006differential}. DP provides mathematically proven privacy protection, ensuring that individual data cannot be extrapolated backward. Additionally, it is possible to quantify privacy loss, striking a balance between data utility and data privacy.

\begin{definition}[\bf{Classical Differential Privacy}] A randomized function $K$ satisfies $(\epsilon,\delta)$-differential privacy if the data sets $D$ and $D^{'}$ differ by only one participant, and all set of outcomes $\mathcal{S} \subseteq Range(\mathcal{M})$ satisfy 
    \begin{equation}
		\rm{Pr}\left[\mathcal{K}(D)\in \mathcal{S}\right] \leq e^\epsilon \cdot \rm{Pr}\left[\mathcal{K}(D')\in \mathcal{S}\right]  + \delta,
    \end{equation}
where $\epsilon$ is the privacy budget and $\delta$ is the broken probability.
    \label{cdp}
\end{definition}

QDP is firstly proposed by Zhou et al.~\cite{zhou2017differential}, which is applicable to quantum centralized computing based on Eq.(\ref{cdp}). QDP incorporates quantum-specific features while retaining the advantages of the classical DP, and finally achieves the effect of protecting quantum data privacy.
\begin{definition}[\bf{Quantum Differential Privacy}] Given two quantum datasets $\rho$ and $\sigma$ with $\tau(\rho,\sigma)\leq d$, where $d \in (0,1]$ and $\tau(\cdot,\cdot)$ indicates the trace distance. A quantum operation $\mathcal{E} $ is $(\epsilon,\delta)$-differentially private if every $POVM$  (Positive Operator-Valued Measure) $M=\left \{ M_{m}  \right \} $ and all~$\mathcal{S} \subseteq Range(\mathcal{M})$ satisfy, 
 \begin{equation}
        \rm{Pr}\left[\mathcal{E} \left ( \rho \right ) \in _{M}  S\right] \leq e^\epsilon \cdot \rm{Pr}\left[\mathcal{E} \left ( \sigma \right ) \in _{M}  S\right] + \delta,
 \end{equation}
 where $\epsilon$ is the privacy budget and $\delta$ is the broken probability.
\label{qdp}
\end{definition}

Specifically, $\epsilon$ is the level of privacy protection. The smaller the $\epsilon$, the stronger the protection. $\delta$ is the probability that the model will obey  $\epsilon$-DP. The smaller the $\delta$, the better the performance.

\subsection{Quantum distributed computing}
Monolithic quantum computing refers to the execution of quantum algorithms on a single quantum processor. In the NISQ era, the computational process of a quantum computer generates inherent noise such as decoherence noise, shot noise, etc. This makes it necessary for quantum circuits to be fault-tolerant, i.e., a logical qubit is represented by multiple physical qubits to reduce error rates. Therefore, a quantum processor with only a few hundred qubits cannot meet the demands of quantum computing in a short period. However, academia and industry agree that a good solution to this problem is to use QDC~\cite{barral2024review}.

QDC is based on a network infrastructure that brings together small and medium-sized quantum chips, thereby increasing the number of available qubits and enabling exponential growth in the computing power of quantum computers. Recently, IBM introduced the Quantum Distributed System Two~\cite{gambe2023ibm,barral2024review}, which integrates three IBM Quantum Heron processors, each featuring 133 fixed-frequency qubits with tunable couplers. IBM claimed that their system offers a 3-5x performance increase over the 127-qubit Eagle processors.

There are two main aspects of scaling from local to distributed quantum processors: quantum partitioning and execution management~\cite{ferrari2023modular}. Quantum partitioning involves splitting a single quantum algorithm into multiple quantum processors. To make DQC efficient and economical, a quantum compiler needs to find the optimal decomposition and then perform intelligent remote scheduling to reduce the loss caused by EPR pairs. 
EPR pair is a pair of entangled quantum particles whose quantum states are interdependent, enabling correlated measurements regardless of the distance between them. To address this problem, Ferrari et al.~\cite{ferrari2023modular} proposed a general DQC quantum compilation framework and remote scheduling strategy based on how to rationally utilize the location of Telegate or Teledata for different types of circuits. Execution management refers to the partitioning of a given set of executed quantum circuits $P$ into non-overlapping subsets $P_{i}$ that satisfy $P=U_{i}P_{i}$. Each subset is assigned to a different QPU one by one. As shown in Fig.~\ref{fig:execution}, for each round i, there exists a scheduling function $S(i)$ that maps the quantum circuits to a quantum network. If DQC is supported, these quantum circuits will be split into subcircuits, each corresponding to a different QPU.

\begin{figure}[t]
\centering
\includegraphics[width=0.45\textwidth]{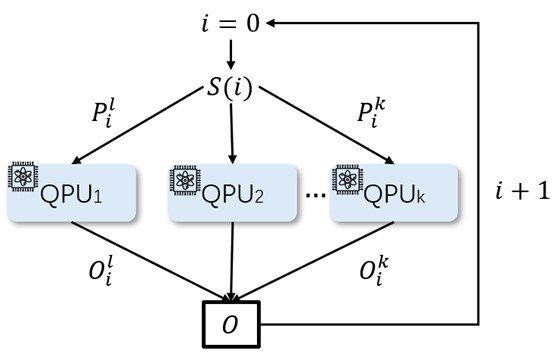}
\caption{Execution of multiple quantum circuit instances with $k$ QPUs.}\label{fig:execution}
\vspace{-4mm}
\end{figure}

\subsection{Classical Rényi differential privacy for classical distributed computers}
We know the privacy-preserving methods in classical and quantum centralized computing through the above. Now, we focus on the privacy-preserving techniques in classical and quantum distributed computing.

In classical distributed computers, RDP is a generalization of classical DP that defines the privacy budget through Rényi divergence, which refers to the difference between two probability distributions~\cite{mironov2017renyi}. RDP controls the privacy budget in a more finely way by introducing a new parameter $\alpha$. It is worth noting that RDP is more suitable for distributed computing, because the total privacy budget can be controlled by $\alpha$ to prevent the explosion of DP compositions. The Rényi divergence and $(\alpha , \epsilon)$-RDP are defined as follows.

\begin{definition}[\bf{Rényi Divergence}] For two probability distribution $P$ and $Q$ defined over $R$, the Rényi divergence of order $\alpha > 1$ is 
 \begin{equation}
        D_{\alpha } (P||Q)=\frac{1}{\alpha-1} logE_{x\sim Q}\left(\frac{P(x)}{Q(x)} \right)^{\alpha}.
 \end{equation}
\label{crenyidivergence}
\end{definition}

\begin{definition}[\bf{$(\alpha , \epsilon)$-RDP}] A randomized mechanism $f: D\mapsto R$ is said to have $\epsilon$-Rényi differential private of order $\alpha$, or $(\alpha , \epsilon)$-QRDP for short, if for any adjacent $D, D' \in \mathcal{D}$, it holds that
    \begin{equation}
		D_{\alpha } (f(D)||f(D'))\le \epsilon.
    \end{equation}
\label{crdp}
\end{definition}

However, when we turned on attention to privacy-preserving approaches in QDC, we found that no existing research has yet studied the quantum version of RDP. This gap prompted us to consider a critical question: is it possible to design a new RDP framework tailored specifically for QDC that can effectively protect data privacy? Such a framework needs to consider the unique properties of quantum systems, while still providing robust privacy guarantees similar to classical RDP.

\section{Differential Privacy Concept for Quantum Distributed Computing\label{sec:QRDP}}
Inspired by the generalization of DP to QDP for the quantum environment, we propose to develop quantum Rényi Differential Privacy (QRDP) based on RDP, which is particularly suitable for more complex scenarios such as QDC, to provide more flexible privacy protection for quantum computing. A key advantage of QRDP is its composition property, thus allowing flexibility in controlling and allocating privacy budgets. For instance, in the QDC architecture illustrated in Fig.~\ref{fig:execution}, we focus on each round and hope to obtain the privacy budget accumulated by multiple quantum processors working together through QRDP. In this section, we define the QRDP concept based on the new quantum Rényi divergence and derive QRDP basic theorems and properties, including post-processing and composition.

\subsection{Quantum Rényi differential privacy}

We formally define quantum Rényi divergence based on classical Rényi divergence. We need to change the two probability distributions $P$ and $Q$ in Eq.(\ref{crenyidivergence}) to correlate with the results of quantum computation. We find that if the quantum states are measured several times, the results are presented as a distribution. Therefore, we replace $P$ and $Q$ with the results after quantum measurements in quantum circuits. Our quantum Rényi divergence is defined as follows.

\begin{definition}[\bf{Quantum Rényi Divergence}] Given two quantum datasets $\rho$ and $\sigma$ with $\tau(\rho,\sigma)\leq d$, where $d \in (0,1]$ and $\tau(\cdot,\cdot)$ indicates the trace distance. $\mathcal{E}_{m} (\cdot )$ represents the result of a quantum state with a quantum operation $\mathcal{E}$ and being measured by $POVM$ $M={M_{m}}$. For two probability distribution $\mathcal{E}_{m} (\rho)$ and $\mathcal{E}_{m} (\sigma)$, the Rényi divergence of order $\alpha > 1$ is
    \begin{equation}
		D_{\alpha } (\mathcal{E}_{m} (\rho)||\mathcal{E}_{m} (\sigma))=\frac{1}{\alpha-1} logE_{\mathcal{E}_{m}(\sigma)}\left(\frac{\mathcal{E}_{m} (\rho)}{\mathcal{E}_{m} (\sigma)} \right)^{\alpha}.
    \end{equation}
    \label{qrdivergence}
\end{definition}

For the endpoints of interval $(1,\infty )$, quantum Rényi divergence is defined in terms of continuity. Specifically, $D_{1}(\mathcal{E}_{m}(\rho)||\mathcal{E}_{m} (\sigma))$ set to $\lim_{\alpha \to 1} D_{\alpha}(\mathcal{E}_{m}(\rho)||\mathcal{E}_{m} (\sigma))$, and can also be called the Kullback-Leibler divergence:
    \begin{equation}
		D_{1} (\mathcal{E}_{m} (\rho)||\mathcal{E}_{m} (\sigma))= logE_{\mathcal{E}_{m}(\sigma)}\left(\frac{\mathcal{E}_{m} (\rho)}{\mathcal{E}_{m} (\sigma)} \right)^{\alpha}.
    \end{equation}

The relation between quantum Rényi divergence and differential privacy when $\alpha =\infty $ is straightforward. A quantum mechanism satisfies $\epsilon$-QDP if and only if two results distribution of quantum states $\rho$ and $\sigma$ satisfies 
    \begin{equation}
		D_{\infty} (\mathcal{E}_{m} (\rho)||\mathcal{E}_{m} (\sigma))\le \epsilon.
    \end{equation}

Based on the new quantum Rényi divergence, we can now propose a definition of QRDP.

\begin{definition}[\bf{$(\alpha , \epsilon)$-Quantum Rényi Differential Privacy}] $\mathcal{E}_{m} (\cdot )$ is $\epsilon$-differential private of order $\alpha$, or $(\alpha , \epsilon)$-QRDP for short, if every $POVM$ $M=\left \{ M_{m}  \right \} $ and all~$\mathcal{S} \subseteq Range(\mathcal{M})$. For any adjacent $\rho$ and $\sigma$ such that $\tau(\rho,\sigma)\leq d$, it holds that
    \begin{equation}
		D_{\alpha } (\mathcal{E}_{m} (\rho)||\mathcal{E}_{m} (\sigma))\le \epsilon.
    \end{equation}
    \label{qrdp}
\end{definition}

According to Proposition 10 in~\cite{mironov2017renyi}, if $\mathcal{E}_{m} (\cdot )$ is $(\alpha , \epsilon)$-QRDP, we also get this formula:
    \begin{equation}
        \begin{split}
	   e^{-\epsilon} \Pr[\mathcal{E}_{m} (\sigma) \in S]^{\alpha / (\alpha - 1)} \leq \Pr[\mathcal{E}_{m} (\rho) \in S] &\\\leq \left(e^{\epsilon} \Pr[\mathcal{E}_{m} (\sigma) \in S]\right)^{(\alpha - 1)/\alpha}.\\
        \end{split}
    \end{equation}

\subsection{Post-processing and basic composition}

Similar to QDP, QRDP is also unaffected by post-processing. Without knowledge of the redundant information in the private dataset, any quantum operation performed redundantly on the quantum operation with $(\alpha, \epsilon)$-QRDP does not increase the privacy loss. Post-processing is defined as follows.

\begin{proposition}[\bf{Post-Processing}] Let $\mathcal{E}$ be a quantum operation that is $(\alpha , \epsilon)$-quantum Rényi differential privacy. Let $\mathcal{F}$ be an arbitrary quantum operation. Then the composition of $\mathcal{E}$ and $\mathcal{F}$:
    \begin{equation}
		\mathcal{F} \circ \mathcal{E} : \rho \mapsto \mathcal{F}(\mathcal{E}(\rho))
    \end{equation}
is $(\alpha , \epsilon)$-quantum Rényi differential privacy too.
\end{proposition}

\begin{proof}
Reference ~\cite{mironov2017renyi} proved that for a randomized mapping $g: R\mapsto R'$, two probability distribution $P$ and $Q$ satisfies $D_{\alpha } (P||Q)\le D_{\alpha } (g(P)||g(Q))$. Since our definitions of $\mathcal{E}_{m} (\rho)$ and $\mathcal{E}_{m} (\sigma)$ can be seen as two special probability distributions, for an arbitrary quantum operation $\mathcal{F}$, they satisfy 
\begin{align}
    D_{\alpha } (\mathcal{E}_{m} (\rho)||\mathcal{E}_{m} (\sigma))\le D_{\alpha } (F(\mathcal{E}_{m} (\rho))||F(\mathcal{E}_{m} (\sigma))). 
\end{align}
It means that if $\mathcal{E}(\cdot )$ is $(\alpha , \epsilon)$-QRDP, so is $F(\mathcal{E}(\cdot ))$. So QRDP is preserved by post-processing.
\end{proof}

Next, we prove some basic composition theorems for QRDP. These theorems are generalizations of classical RDP to the quantum setting. The composition theorems are particularly important for QDC in that they provide theoretical support for computing the total privacy budget and can protect the overall privacy in multiple quantum operations or complex quantum algorithms. The composition theorems are defined as follows.

\begin{proposition}[\bf{Basic Composition Theorem}] Let $M_{1}(\mathcal{E}_{1m1}): D_{1}\mapsto R_{1}$ be an $(\alpha , \epsilon_{1})$-quantum Rényi differential private algorithm. Let $M_{2}(\mathcal{E}_{2m2}): D_{2}\mapsto R_{2}$ be an $(\alpha , \epsilon_{2})$-quantum Rényi differential private algorithm. Assume that they are independent, their combination
    \begin{equation}
		h:D_{1}\otimes D_{2}\mapsto R_{1}\otimes R_{2}
    \end{equation}
is defined by 
    \begin{equation}
		h(\rho,\sigma)=M_{1}(\rho)\otimes M_{2}(\sigma)
    \end{equation}
for all density operators $\rho$ in $D_{1}$ and $\sigma$ in $D_{2}$. Then $h$ is $(\alpha , \epsilon_{1}+\epsilon_{2})$-quantum Rényi differential private.
\end{proposition}

\begin{proof}
We refer to the idea of Proposition 1 in~\cite{mironov2017renyi} for the proof. Let $\rho_{1},\rho_{2}\in D_{1}$ be neighbor quantum states, and let $\sigma_{1},\sigma_{2}\in D_{2}$ be neighbor quantum states. Eq.(\ref{qrdivergence}) can be
\begin{equation}
\begin{split}
		&exp[(\alpha-1)D_{\alpha}(h(\rho_{1},\sigma_{1})||h(\rho_{2},\sigma_{2}))]\\
&=\sum_{x\in R_{1}} \sum_{y\in R_{2}}p^{\alpha}((\mathcal{E}_{1m1})(\rho_{1})\otimes (\mathcal{E}_{2m2})(\sigma_{1})=x\otimes y)\cdot \\&p^{\alpha-1}(\mathcal{E}_{1M1}(\rho_{2})\otimes \mathcal{E}_{2M2}(\sigma_{2})=x\otimes y)\\
&= \sum_{x\in R_{1}} \sum_{y\in R_{2}}p^{\alpha}(\mathcal{E}_{1M1}(\rho_{1})=x)p^{\alpha}(\mathcal{E}_{2M2}(\sigma_{1})=y)\cdot\\&p^{\alpha-1}(\mathcal{E}_{1M1}(\rho_{2})=x)p^{\alpha-1} (\mathcal{E}_{2M2}(\sigma_{2})=y)\\
&=\sum_{x\in R_{1}} p^{\alpha}(\mathcal{E}_{1M1}(\rho_{1})=x)p^{\alpha-1}(\mathcal{E}_{1M1}(\rho_{2})=x)\cdot\\& \sum_{y\in R_{2}}p^{\alpha}(\mathcal{E}_{2M2}(\sigma_{1})=y) p^{\alpha-1} (\mathcal{E}_{2M2}(\sigma_{2})=y)\\
&\le exp \left((\alpha-1) \epsilon_{1}\right)\exp \left((\alpha-1) \epsilon_{2}\right)\\
&= exp \left((\alpha-1) (\epsilon_{1}+\epsilon_{2})\right).\\
\end{split}
    \end{equation}
Then we proved that $h$ is $(\alpha , \epsilon_{1}+\epsilon_{2})$-quantum Rényi differential private.
\end{proof}

We further derive the specific composition theorem, where mechanisms are related to each other. The proposition is as shown below.

\begin{proposition}[\bf{Associated Composition Theorem}] 
Let $M_{1}$   $(\mathcal{E}_{1m1}): D_{1}\mapsto R_{1}$ be an $(\alpha , \epsilon_{1})$-quantum Rényi differential private algorithm. Let $M_{2}(\mathcal{E}_{2m2}): R_{1}\times D\mapsto R_{2}$ be an $(\alpha , \epsilon_{2})$-quantum Rényi differential private algorithm. Then, the mechanism defined as $(X,Y)$, where $X \leftarrow M_{1}(D)$ and $Y \leftarrow M_{2}(X,D)$ satisfies $(\alpha , \epsilon_{1}+\epsilon_{2})$-quantum Rényi differential privacy.
\end{proposition}
\begin{proof}
Proposition 1 in~\cite{mironov2017renyi} shows that the associated composition theorem holds in classical computing. Since our definitions of $\mathcal{E}_{m} (\rho)$ and $\mathcal{E}_{m} (\sigma)$ can be seen as two special probability distributions, the associated composition theorem also holds in quantum computing.
\end{proof}

\section{Differential Privacy Mechanism for Quantum Distributed Computing\label{sec:mechanism}}
Classical DP protects the user's privacy from being compromised by adding artificial noise to the model, such as Gaussian noise, and Laplace noise. QDP proposed by~\cite{zhou2017differential} promotes this idea by adding quantum noise to the quantum model to protect data privacy in the quantum environment. Next, we analyze how these quantum noise mechanisms will affect QRDP. We assume that the input to the quantum model is the quantum state $\rho$. $\mathcal{E}(\rho)$ denotes the result of that quantum state after the quantum gate $\mathcal{E}$ operation. Then we add a quantum noise $\mathcal{E}_{N}$  to the calculation result, which gets the final result $\mathcal{E}_{N}(\mathcal{E}(\rho))$. We can now define the quantum noise mechanism that implements QRDP as:
\begin{equation}
    \mathcal{E}_N \circ \mathcal{E} : \rho \mapsto \mathcal{E}_N(\mathcal{E}(\rho)).
\end{equation}

We assumed that there are two neighbor quantum states $\rho$ and $\sigma$ which satisfies $\tau (\rho,\sigma)\le d$. For ease of representation, $\rho'=\mathcal{E}(\rho)$, $\sigma'=\mathcal{E}(\sigma)$ represent the result of the quantum state after quantum operation $\mathcal{E}$. Let's take a concrete example. If the quantum state is a $2\times2$ matrix, after a $2\times2$ quantum gate, $\rho'$ and $\sigma'$ are still $2\times2$ density matrices, which can be expressed as

\[
\rho' = \begin{bmatrix} a & b \\ b^* & c \end{bmatrix}, \quad
\delta' = \rho' - \sigma' = \begin{bmatrix} \delta_a & \delta_b \\ \delta_b^* & -\delta_a \end{bmatrix},
\]
\[
\sigma' = \rho'-\delta'=\begin{bmatrix} a-\delta_a & b-\delta_b \\ b^*-\delta_b^* & c+\delta_a \end{bmatrix},
\]
where $c=1-a$ and $\tau (\rho',\sigma')\le d$. We take $\rho'$ as an example. $a$ and $c$ are real numbers denoting the distribution probability of the quantum state. $b$ and $b^{*}$ are complex numbers related to the coherence of the quantum state, and $b^{*}$ is the complex conjugate of $b$. These values are determined when the input quantum state and the quantum operation are determined. Then $\rho''=\mathcal{E}_{N}(\rho')$, $\sigma''=\mathcal{E}_{N}(\sigma')$ represent the result after quantum noise $\mathcal{E}_{N}$. 

In the sequel, we give theorems under different noise mechanisms: generalized amplitude damping, the combination of phase and amplitude damping, and depolarizing mechanism.

\subsection{Generalized amplitude damping}

The generalized amplitude damping (GAD) noise model simulates the evolution of a quantum system as it exchanges energy with the external environment. According to our definition, the GAD noise mechanism used to implement QRDP can be expressed as
\begin{equation}
    \rho \mapsto \mathcal{M}_{\text{GAD}}(\rho) = \mathcal{E}_{\text{GAD}}[\mathcal{E}(\rho)].
\end{equation}

We take a more concrete example to understand the calculation of this mechanism. Assuming that the GAD noise acts on a single qubit, the GAD channel can be expressed as
\begin{equation}
   \mathcal{E}_{GAD}(\rho)=\sum_{k=0}^{3} E_{k}\rho E_{k}^{\dagger },
   \label{gad}
\end{equation}
in the two-dimensional Hilbert space, where 
\[
E_0 = \sqrt{p} \begin{bmatrix} 1 & 0 \\ 0 & \sqrt{1-\gamma} \end{bmatrix}, \quad
E_1 = \sqrt{p} \begin{bmatrix} 0 & \sqrt{\gamma} \\ 0 & 0 \end{bmatrix}, 
\]
\[
E_2 = \sqrt{1-p} \begin{bmatrix} \sqrt{1-\gamma} & 0 \\ 0 & 1 \end{bmatrix}, \quad
E_3 = \sqrt{1-p} \begin{bmatrix} 0 & 0 \\ \sqrt{\gamma} & 0 \end{bmatrix}.
\]
$E_0$ and $E_1$ mainly describe the energy exchange of qubits with the external environment in a low-temperature environment; $E_2$ and $E_3$ mainly describe in a high-temperature environment. $p$ and $\gamma$ are two parameters. $p$ denotes the probability that a qubit exchanges energy with the environment. $\gamma$ describes the strength of the coupling between the qubits and the environment~\cite{cafaro2014approximate}. To simplex the mechanism as~\cite{zhou2017differential}, we assume that $p=0.5$. Then $\rho''$ and $\sigma''$ can be expressed as

\[
\rho'' = \mathcal{E}_{\text{GAD}}(\rho') = \begin{bmatrix} 
a + \frac{1}{2}(c-a)\gamma & b\sqrt{1-\gamma} \\
b^*\sqrt{1-\gamma} & c + \frac{1}{2}(a-c)\gamma 
\end{bmatrix},
\]


\[
\begin{aligned}
&\sigma'' = \mathcal{E}_{\text{GAD}}(\sigma')\\
&= \begin{bmatrix} 
a + \frac{1}{2}(c-a)\gamma -\delta_a(1-\gamma) & b\sqrt{1-\gamma} -\delta_b\sqrt{1-\gamma}\\
b^*\sqrt{1-\gamma}-\delta_b^*\sqrt{1-\gamma} & c + \frac{1}{2}(a-c)\gamma +\delta_a(1-\gamma)
\end{bmatrix}.
\end{aligned}
\]

Now we can obtain the following new QRDP theorem under GAD noise mechanism.

\begin{theorem}[\bf{QRDP under Generalized Amplitude Damping}]
For all inputs $\rho$ and $\sigma$ such that $\tau(\rho,\sigma)\le d $, the generalized amplitude damping $M_{GAD}$ in the D-dimension Hilbert space provides $(\alpha,\epsilon)$-quantum Rényi differential private,
\begin{equation}
\begin{split}
   &\frac{1}{\alpha-1} \log E_{\mathcal{E}_{m}(\sigma)}\left(\frac{\mathcal{E}_{m} (\rho)}{\mathcal{E}_{m} (\sigma)} \right)^{\alpha}\\
    &=\frac{1}{\alpha-1} \log \sum_{m} \frac{[Tr(M(\rho'')]^{\alpha }}{[Tr(M(\sigma'')]^{\alpha-1}} \\
    &\le \epsilon,
    \label{gadeps}
\end{split}
\end{equation}
where $\rho''=\mathcal{E}_{GAD}(\rho')$, $\sigma''=\mathcal{E}_{GAD}(\sigma')$, $M=M_{m}$ is the measurement matrix. 

\end{theorem}

When the GAD noise mechanism acts on a single qubit, $\rho''$ and $\sigma''$ can be brought into our example.

\begin{proof}
We express $\mathcal{E}_{m} (\cdot)$ in the form of a matrix calculation, including quantum operation and quantum measurement. Then, Eq.(\ref{qrdivergence}) can be expanded as
\begin{equation}
\begin{split}
    &\frac{1}{\alpha-1} \log E_{\mathcal{E}_{m}(\sigma)}\left(\frac{\mathcal{E}_{m} (\rho)}{\mathcal{E}_{m} (\sigma)} \right)^{\alpha}\\
    &=\frac{1}{\alpha-1} \log \sum_{m} \left(\frac{P(\mathcal{E}_{m} (\rho)=m)}{P(\mathcal{E}_{m} (\sigma)=m)} \right)^{\alpha} \cdot P(\mathcal{E}_{m} (\sigma)=m)\\
    &=\frac{1}{\alpha-1} \log \sum_{m} \frac{P(\mathcal{E}_{m} (\rho)=m)^{\alpha }}{P(\mathcal{E}_{m} (\sigma)=m)^{\alpha-1}}\\
    &=\frac{1}{\alpha-1} \log \sum_{m} \frac{[Tr(M(\rho'')]^{\alpha }}{[Tr(M(\sigma'')]^{\alpha-1}} \\
    &\le \epsilon.
\end{split}
\end{equation}
\end{proof}

\subsection{Composition of phase and amplitude damping}
The composition of phase and amplitude damping (PAD) noise model describes qubits as being affected by the loss of phase information while undergoing energy decay. According to our definition, the PAD noise mechanism used to implement QRDP can be expressed as
\begin{equation}
    \rho \mapsto \mathcal{M}_{\text{PAD}}(\rho) = \mathcal{E}_{\text{GAD}}[\mathcal{E}_{\text{PD}}(\mathcal{E}(\rho))].
    \label{pad}
\end{equation}

We take a more concrete example. Assuming that the PD noise acts on a single qubit, the PD channel can be expressed as
\begin{equation}
   \mathcal{E}_{PD}(\rho)= E_{0}\rho E_{0}^{\dagger }+ E_{1}\rho E_{1}^{\dagger },
\end{equation}
in the two-dimensional Hilbert space, where 
\[
E_0 = \begin{bmatrix} 
1 & 0 \\
0 & \sqrt{1-\lambda} 
\end{bmatrix}, \quad
E_1 = \begin{bmatrix} 
0 & 0 \\
0 & \sqrt{\lambda} 
\end{bmatrix}.
\]
$E_0$ mainly describes the phase information of a qubit, but may result in a decrease in the amplitude of the quantum state (related to $\lambda $). $E_1$ results in a loss of phase of the quantum state, but does not change the energy of the quantum state. $\lambda$ represents the probability that the phase information will be lost. In this case, $\rho'$ and $\sigma'$ after phase damping noise can be expressed as

\[
\rho_{PD}''=\mathcal{E}_{\text{PD}}(\rho') = \begin{bmatrix} 
a & b\sqrt{1-\lambda} \\
b^*\sqrt{1-\lambda} & c
\end{bmatrix},
\]

\[
\begin{aligned}
\sigma_{PD}''&=\mathcal{E}_{\text{PD}}(\sigma')\\ &= \begin{bmatrix} 
a-\delta_{a} & (b-\delta_{b})\sqrt{1-\lambda} \\
(b^*-\delta_{b^*})\sqrt{1-\lambda} & c+\delta_{a}
\end{bmatrix}.
\end{aligned}
\]

Then, we can get final $\rho''$ and $\sigma''$ after the composition of phase and amplitude damping noise:
\[
\rho'' = \mathcal{E}_{\text{GAD}}(\rho_{PD}'') = \begin{bmatrix} 
a + \frac{1}{2}(c-a)\gamma & b\sqrt{1-\gamma}\sqrt{1-\lambda} \\
b^*\sqrt{1-\gamma}\sqrt{1-\lambda} & c + \frac{1}{2}(a-c)\gamma 
\end{bmatrix},
\]

\[
\begin{aligned}
&\sigma'' = \mathcal{E}_{\text{GAD}}(\sigma_{PD}'') \\
&= \begin{bmatrix} 
a + \frac{1}{2}(c-a)\gamma -\delta_a(1-\gamma) & (b-\delta_b)\sqrt{1-\gamma}\sqrt{1-\lambda}\\
(b^*-\delta_b^*)\sqrt{1-\gamma}\sqrt{1-\lambda} & c + \frac{1}{2}(a-c)\gamma +\delta_a(1-\gamma)
\end{bmatrix}.
\end{aligned}
\]

Now we can obtain the following new QRDP theorem under PAD noise mechanism.

\begin{theorem}[\bf{QRDP under Composition of Phase and Amplitude Damping}]
For all inputs $\rho$ and $\sigma$ such that $\tau(\rho,\sigma)\le d $, the composition of phase and amplitude damping $M_{PAD}$ in the D-dimension Hilbert space provides $(\alpha,\epsilon)$-quantum Rényi differential private,
\begin{equation}
\begin{split}
   &\frac{1}{\alpha-1} \log E_{\mathcal{E}_{m}(\sigma)}\left(\frac{\mathcal{E}_{m} (\rho)}{\mathcal{E}_{m} (\sigma)} \right)^{\alpha}\\
    &=\frac{1}{\alpha-1} \log \sum_{m} \frac{[Tr(M(\rho'')]^{\alpha }}{[Tr(M(\sigma'')]^{\alpha-1}} \\
    &\le \epsilon,
    \label{padeps}
\end{split}
\end{equation}
where $\rho''=\mathcal{E}_{PAD}(\rho')$, $\sigma''=\mathcal{E}_{PAD}(\sigma')$, $M=M_{m}$ is the measurement matrix. 

\end{theorem}

The proof is similar to Theorem 1, with the difference being the $\rho''$ and $\sigma''$ due to different noise. When the PAD noise mechanism acts on a single qubit, $\rho''$ and $\sigma''$ can be brought into our example.



\subsection{Depolarizing mechanism}
The depolarizing (Dep) noise model describes the randomization of a qubit due to environmental disturbances during transmission or evolution. When a qubit is affected by noise, it has a certain probability of losing its original quantum information and becoming completely mixed. According to our definition, the Dep noise mechanism used to implement QRDP can be expressed as
\begin{equation}
    \rho \mapsto \mathcal{M}_{\text{Dep}}(\rho) = \mathcal{E}_{\text{Dep}}[\mathcal{E}(\rho)].
\end{equation}

The depolarizing noise can be expressed in terms of quantum operations as
\begin{equation}
    \mathcal{E}_{\text{Dep}}(\rho) = \frac{pI}{D} + (1 - p)\rho,
    \label{dep}
\end{equation}
where $\rho$ is the initial quantum state, $I$ is the unit operator, $p$ is the depolarizing noise error rate, and $D$ is the dimension of the Hilbert space. Within the depolarizing noise channel, the quantum state remains unchanged with a probability of $1-p$. However, with a probability of $\frac{p}{D}$, it transitions uniformly to any possible quantum state.

We take a more concrete example to understand the calculation of this mechanism. Assuming that the Dep noise acts on a single qubit, $\rho''$ and $\sigma''$ can be expressed as 
\[
\rho'' = \mathcal{E}_{\text{Dep}}(\rho') = \begin{bmatrix} 
\frac{pI}{D} + (1 - p)a & b(1-p) \\
b^*(1-p) & \frac{pI}{D} + (1 - p)c 
\end{bmatrix},
\]

\[
\begin{aligned}
&\sigma'' = \mathcal{E}_{\text{Dep}}(\sigma') \\
&= \begin{bmatrix} 
\frac{pI}{D} + (1 - p)(a-\delta_{a}) & (b-\delta_{b})(1-p) \\
(b^*-\delta_{b^*})(1-p) & \frac{pI}{D} + (1 - p)(c+\delta_{a}) 
\end{bmatrix}.
\end{aligned}
\]

Now we can obtain the following new QRDP theorem under Dep noise mechanism.

\begin{theorem}[\bf{QRDP under Depolarizing Mechanism}]
For all inputs $\rho$ and $\sigma$ such that $\tau(\rho,\sigma)\le d $, the depolarizing mechanism $M_{Dep}$ in the D-dimension Hilbert space provides $\hat{\epsilon}(\alpha)$-quantum Rényi differential private,
\begin{equation}
\begin{split}
   &\frac{1}{\alpha-1} \log E_{\mathcal{E}_{m}(\sigma)}\left(\frac{\mathcal{E}_{m} (\rho)}{\mathcal{E}_{m} (\sigma)} \right)^{\alpha}\\
    &=\frac{1}{\alpha-1} \log \sum_{m} \frac{[Tr(M_{m}(\rho'')]^{\alpha }}{[Tr(M_{m}(\sigma'')]^{\alpha-1}} \\
    &\le \epsilon,
    \label{depeps}
\end{split}
\end{equation}
where $\rho''=\mathcal{E}_{Dep}(\rho')$, $\sigma''=\mathcal{E}_{Dep}(\sigma')$, $M=M_{m}$ is the measurement matrix. 

\end{theorem}

The proof is similar to Theorem 1, with the difference being the $\rho''$ and $\sigma''$ due to different noise. When the Dep noise mechanism acts on a single qubit, $\rho''$ and $\sigma''$ can be brought into our example.

\subsection{An intuitive QRDP under different noise mechanisms}
From Eq.(\ref{gadeps}) - (\ref{depeps}), we derive the lowest privacy budget under different noise mechanisms since we do not scale the formulas during the computation, thereby allowing for a more precise estimation of the privacy budget. However, we found that at this point the privacy budget is related to numerous parameters, including the input quantum state and the measurement matrix. To generalize this, we now propose a more comprehensive privacy budget calculation that not only serves as an upper bound for the privacy budget obtained across different input quantum states and measurements but also provides a clearer understanding of how different noise parameters affect the privacy budget. By relaxing the formulation, we introduce a more intuitive theorem that focuses on the privacy budget derived solely from the quantum noise parameters. Both ways of calculating the privacy budget are valid, and the readers can choose the appropriate QRDP formula based on their specific application requirements. The intuitive QRDP theorem is as follows.

\begin{theorem}[\bf{Intuitive QRDP under Different Noise Mechanisms}\label{intuitiveqrdp}]
For all inputs $\rho$ and $\sigma$ such that $\tau(\rho,\sigma)\le d $, the depolarizing mechanism $M$ in the D-dimension Hilbert space provides $\hat{\epsilon}(\alpha)$-quantum Rényi differential private,
\begin{equation}
\begin{split}
   \hat{\epsilon}(\alpha)
   &=\epsilon-\frac{1}{\alpha-1} \log \left(\frac{1+e^{-\epsilon}}{1+e^{-(2 \alpha-1) \epsilon}}\right),
   \label{qdpqrdp}
\end{split}
\end{equation}
where $\epsilon= \ln \left[ 1 + \frac{2d \sqrt{1 - \gamma}}{1 - \sqrt{1 - \gamma}} \right]$ under the generalized amplitude damping; $\epsilon= \ln \left[ 1 + \frac{2d \sqrt{1 - \gamma}\sqrt{1 - \lambda}}{1 - \sqrt{1 - \gamma}\sqrt{1 - \lambda}} \right]$ under the composition of phase and amplitude damping; $\epsilon=\ln_{}{\left [ 1+\frac{1-p}{p}dD \right ] }$ under the depolarizing mechanism.
\end{theorem}

By Theorem~\ref{intuitiveqrdp}, we find that the privacy budget of the QRDP decreases when the quantum noise parameter increases, and this property is independent of the input quantum state. We will illustrate this trend in Section~\ref{sec:simulation}.

\begin{proof} 
First, Ref.~\cite{DPorg-pdp-to-zcdp} gives a way for calculating RDP in classical computing. Let $M: D_{n}\mapsto R$ be a randomized algorithm satisfying $\epsilon$-differential privacy. Then $M$ satisfies $\hat{\epsilon }(\alpha ) $-Rényi DP for all $\alpha > 1$,
\begin{equation}
\begin{split}
   \hat{\epsilon}(\alpha) &=\frac{1}{\alpha-1} \log \left(\frac{1}{e^{\epsilon}+1} e^{\alpha \epsilon}+\frac{e^{\epsilon}}{e^{\epsilon}+1} e^{-\alpha \epsilon}\right) \\
   &=\epsilon-\frac{1}{\alpha-1} \log \left(\frac{1+e^{-\epsilon}}{1+e^{-(2 \alpha-1) \epsilon}}\right),\\
   \label{dprdp}
\end{split}
\end{equation}
where $\epsilon$ is the privacy budget obtained using DP in the corresponding classical noise mechanism. This bound is tight. Equation (\ref{dprdp}) is the Theorem 4 from~\cite{DPorg-pdp-to-zcdp} and the detailed proof is provided in~\cite{DPorg-pdp-to-zcdp}. Equation (\ref{qdpqrdp}) builds on Eq.(\ref{dprdp}). Since our definitions of $\mathcal{E}_{m} (\rho)$ and $\mathcal{E}_{m} (\sigma)$ can be seen as two special probability distributions, the similar theorem also holds in the quantum scenario. The difference is that $\epsilon$ of Eq.(\ref{qdpqrdp}) are referred to QDP under different quantum noise mechanisms. Reference~\cite{zhou2017differential} derived the $\epsilon$ under the three noise mechanisms (Theorems 1, 2, and 3), e.g. under the depolarizing mechanism, $\epsilon=\ln_{}{\left [ 1+\frac{1-p}{p}dD \right ] }$, so we can substitute directly into Eq.(\ref{qdpqrdp}). Then, $\hat{\epsilon }(\alpha ) $ is independent of the quantum circuit, simplifying the calculation.
\end{proof}

\subsection{Data utility under different noise mechanisms}
One advantage of QRDP is that it balances data privacy and data utility. Therefore,  we must pay attention to the relationship between privacy budget and data utility under different noise mechanisms. We reflect data utility through fidelity, which provides a formula to quantify the degree of similarity of a pair of quantum states~\cite{liang2019quantum}. Since qubits are susceptible to noise, the researchers always use fidelity to reflect the closeness of the actual state to the expected state. The most widely used is the Schumacher fidelity~\cite{jozsa1994fidelity,uhlmann1976transition}, which can represent the maximum transfer probability between a pair of density matrices $\rho$ and $\sigma$ purity. The Schumacher's fidelity is calculated as
\begin{equation}
    \mathcal{F}(\rho, \sigma) := \max_{\{|\psi\rangle, |\phi\rangle\}} |\langle \psi | \phi \rangle|^2 = \left(\mathrm{Tr}\left(\sqrt{\sqrt{\rho} \sigma \sqrt{\rho}}\right)\right)^2.
\end{equation}

We use this fidelity to reflect the degree of similarity between the quantum state $\rho'$ without noise and the quantum state $\rho''$ with noise, thus reflecting the utility of the data. The closer the fidelity is to 1, the more similar the noisy quantum state is to the noiseless quantum state, indicating better data utility. The formula is
\begin{equation}
    \mathcal{F}(\rho', \rho'') := \max_{\{|\psi\rangle, |\phi\rangle\}} |\langle \psi | \phi \rangle|^2 = \left(\mathrm{Tr}\left(\sqrt{\sqrt{\rho'} \rho'' \sqrt{\rho'}}\right)\right)^2.
    \label{fid}
\end{equation}

We can now demonstrate the relationship between fidelity and the quantum noise parameter in the following theorem.

\begin{theorem}[\bf{Fidelity Trends under Different Noise Mechanisms}\label{fidelitytrends}]
Let $\rho' = \begin{bmatrix} a & b \\ b^* & c \end{bmatrix}$ be a $2\times 2$ quantum density matrix (with $a+c=1$ and $b \in \mathbb{C}$) and let $\rho''$ be the density matrix under GAD, PAD, Dep noise mechanisms (the specific forms of $\rho''$ can be found in Section 4.1-4.3). The fidelity $\mathcal{F}(\rho', \rho'')$ between $\rho'$ and $\rho''$ is a decreasing function of the quantum noise parameters.
\end{theorem}

The detailed proof can be found in the Appendix. Theorem~\ref{fidelitytrends} illustrates that as the parameter of the quantum noise increases, the fidelity decreases, independent of the input quantum state. We will use this theorem in Section~\ref{sec:simulation} to represent the relationship between data utility and data privacy.

\section{Quantum RDP and $(\epsilon ,\delta )$-QDP\label{sec:QRDPqdp}}
In classical computing, the privacy budget obtained by the DP composition is too loose to reflect an accurate privacy-preserving capability; while RDP composition provides a much tighter total privacy budget~\cite{murtagh2015complexity}. However, it also makes RDP lack operational interpretation, i.e., not as intuitive and easy to understand as DP~\cite{balle2020hypothesis,asoodeh2021three}. Thus Proposition 3 in~\cite{mironov2017renyi} provides a transformation from $(\alpha,\epsilon)$-RDP to $(\epsilon,\delta)$-DP while retaining the advantages of both DP and RDP. We therefore propose an extension of the quantum domain on this basis: from $(\alpha,\epsilon)$-QRDP to $(\epsilon,\delta)$-QDP.

We mentioned in Section~\ref{sec:QRDP} that when $\alpha$ equals infinity, i.e., $(\alpha,\epsilon)$-QRDP, it is equivalent to $\epsilon$-QDP. By the monotonicity of the quantum Rényi divergence, $(\infty,\epsilon)$-QRDP satisfies $(\alpha, \epsilon)$-QRDP, where $\alpha$ is finite. In turn, $(\alpha,\epsilon)$-QRDP can also satisfy $(\epsilon_{\delta},\delta)$-QDP, where $\delta> 0$. The transformation from QRDP to QDP is as follows.

\begin{proposition}[\bf{From Quantum RDP to $(\epsilon ,\delta )$-QDP}] If $\mathcal{E}_{m} (\cdot )$ is an $(\alpha , \epsilon)$-QRDP mechanism, it also satisfies $(\epsilon + \frac{log1/\delta }{\alpha -1}  ,\delta )$-quantum differential privacy for any $0< \delta < 1$.
\end{proposition} 

\begin{proof} 
Proposition 3 in~\cite{mironov2017renyi} shows that if $f$ is an $(\alpha , \epsilon)$-RDP mechanism, it also satisfies $(\epsilon + \frac{log1/\delta }{\alpha -1}  ,\delta )$-differential privacy for any $0< \delta < 1$. $f$ refers to the general probability distribution. It also provides detailed proof. Since our definitions of $\mathcal{E}_{m} (\cdot)$ can be seen as a special probability distribution applicable in the quantum setting, conversion from RDP and $(\epsilon ,\delta )$-DP also holds in quantum computing.
\end{proof}

\section{Performance Evaluation\label{sec:simulation}}
{We study the relationship between the quantum noise parameters and the $(\alpha,\epsilon)$-QRDP under different quantum noise mechanisms. As analyzed above, we have presented two ways for computing the QRDP. In our simulations, we choose Theorem 5 to calculate the privacy budget. The reason is that Theorem 5 is independent of specific quantum states and quantum measurements, and is only related to the quantum noise parameters, which makes it more suitable for analyzing the impacts of quantum noises on the QRDP. That is similar to the approach used in~\cite{zhou2017differential}. The trace distance $d$ is fixed as long as the dataset and the encoding method are determined. 

\subsection{$(\alpha,\epsilon)$-QRDP under different mechanisms.}

According to Theorems 1, 2, 3 mentioned in~\cite{zhou2017differential}, the privacy budget for QDP is independent of quantum circuits and quantum measurements. Under the generalized amplitude damping mechanism, p is fixed ($p=0.5$), when $\gamma$ increases, $\epsilon$ decreases; under the composition of phase and amplitude damping mechanism, p is fixed ($p=0.5$), when $\gamma$ or $\lambda$ increases, $\epsilon$ decreases; under the depolarizing mechanism, when $p$ increases, $\epsilon$ decreases~\cite{zhou2017differential}. Therefore, for QRDP, we will investigate whether the relationship between $\epsilon(\alpha)$ and these parameters change for these three noise mechanisms under the same conditions as~\cite{zhou2017differential}. We assume the trace distance between two neighboring quantum states $\rho$ and $\sigma$ is $d = 0.1$.

\textbf{Generalized amplitude damping.}
According to and Eq.(\ref{gad}) and (\ref{qdpqrdp}), under the GAD mechanism, $(\alpha,\epsilon)$-QRDP is related to two parameters, $p$ and $\gamma$. To simplex the mechanism as~\cite{zhou2017differential}, we assume that $p=0.5$. Figure~\ref{fig:gad_gamma} shows that when $\gamma$ increases, the privacy budget decreases, leading to stronger privacy protection. $\gamma$ increase indicates that the environmental interference becomes stronger, making it harder to extract information from the system. In addition, it also shows that the privacy budget increases as $\alpha$ increases. $\alpha$ increase indicates that the model focuses more on the most extreme cases, i.e. the worst case for privacy breach.

\begin{figure}[t]
    \centering  \includegraphics[width=0.45\textwidth]{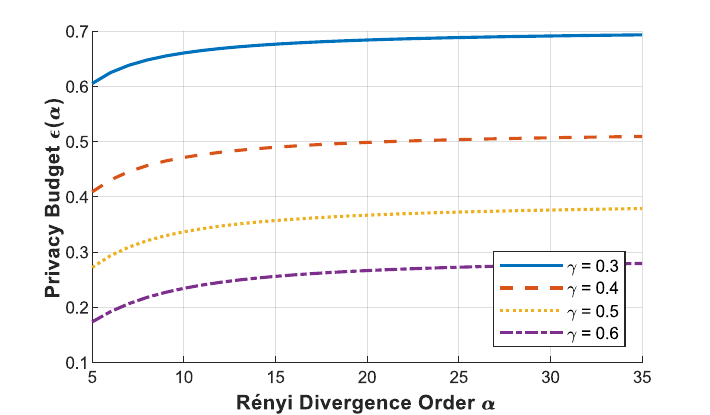}
    \caption{Relationship between $\epsilon$ and $\alpha$ under generalized amplitude damping $\gamma$.}
    \label{fig:gad_gamma}
    \vspace{-2mm}
\end{figure}



\textbf{Composition of phase and amplitude damping.}
According to Eq.(\ref{pad}) and (\ref{qdpqrdp}), under the PAD mechanism, $(\alpha,\epsilon)$-QRDP is related to three parameters, $p$, $\gamma$ and $\lambda$. To simplex the mechanism as~\cite{zhou2017differential}, we assume that $p=0.5$. Throughout the simulation, the circuit parameters are set to $\gamma = 0.3$ and $\lambda=0.2$ unless otherwise stated. All parameters take values between 0 and 1. Figure~\ref{fig:pad_gamma} shows that when $\gamma$ increases, the privacy budget decreases and privacy protection becomes better. $\gamma$ increase indicates that the environmental interference becomes stronger. Figure~\ref{fig:pad_lambda} shows that when $\lambda$ increases, the privacy budget also decreases, resulting in enhanced privacy protection. $\lambda$ increase indicates that the quantum state loses more phase information and introduce more noise. In addition, these two graphs also show that the privacy budget increases as $\alpha$ increases.

\begin{figure}[t]
    \centering
    \includegraphics[width=0.45\textwidth]{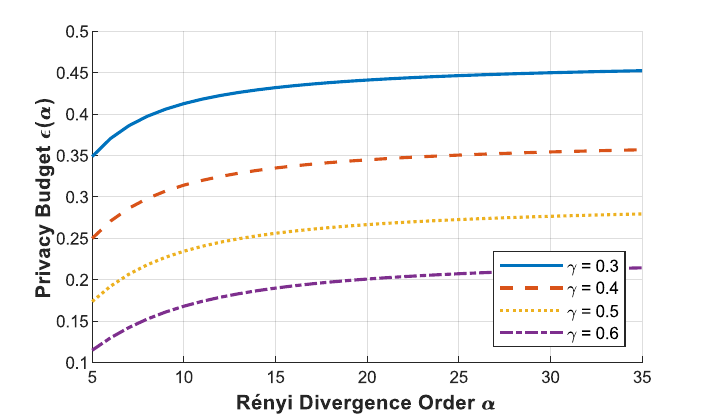}
    \caption{Relationship between $\epsilon$ and $\alpha$ under composition of phase and amplitude damping $\gamma$.}
    \label{fig:pad_gamma}
    \vspace{-2mm}
\end{figure}

\begin{figure}[t]
    \centering
    \includegraphics[width=0.45\textwidth]{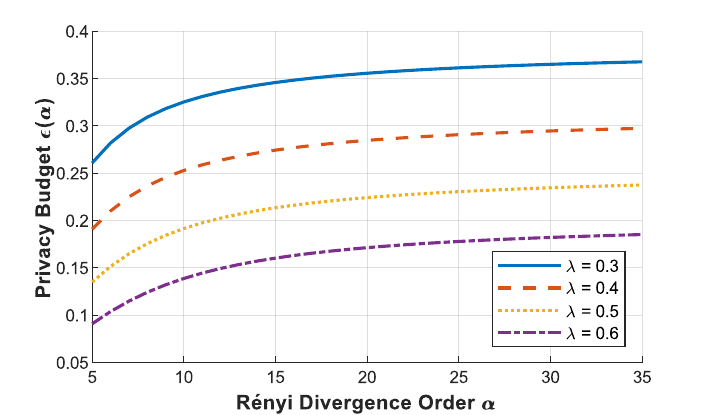}
    \caption{Relationship between $\epsilon$ and $\alpha$ under composition of phase and amplitude damping $\lambda$.}
    \label{fig:pad_lambda}
     \vspace{-2mm}
\end{figure}

\textbf{Depolarizing mechanism.}
According to Eq.(\ref{dep}) and (\ref{qdpqrdp}), under the Dep mechanism, when the quantum dataset and circuit are determined, $D$ is fixed ($D=2$ for the single qubit input), so $(\alpha,\epsilon)$-QRDP is only related to $p$. Figure~\ref{fig:dep_p} shows that when $p$ increases, the privacy budget decreases, further improving privacy protection. $p$ increase indicates that quantum states are more likely to become fully mixed states. A fully mixed state is a state with minimum information, so quantum states have minimum differences and reduce information leakage. In addition, it also shows that the privacy budget increases as $\alpha$ increases.

\begin{figure}[t]
    \centering
    \includegraphics[width=0.45\textwidth]{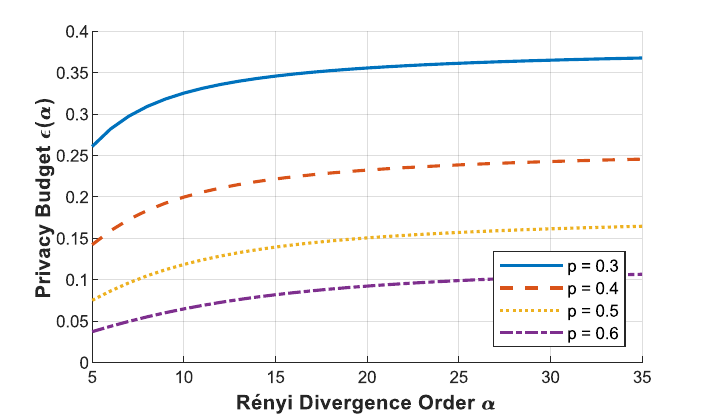}
    \caption{Relationship between $\epsilon$ and $\alpha$ under depolarizing mechanism $p$.}
    \label{fig:dep_p}
     \vspace{-2mm}
\end{figure}

\subsection{Fidelity under different mechanisms.}
Based on Theorems~\ref{intuitiveqrdp} and~\ref{fidelitytrends}, we rigorously demonstrate that the trends in QRDP and fidelity are independent of quantum states. Therefore we randomly take two quantum states $\rho'$ and $\sigma'$ to calculate the corresponding fidelity and privacy budgets,
\[
\rho' = \begin{bmatrix} 0.3 & 0.2 \\ 0.2 & 0.7 \end{bmatrix}, \quad
\sigma' = \begin{bmatrix} 0.4 & 0.1 \\ 0.1 & 0.6 \end{bmatrix}.
\]
Their trace distance is approximately $d=0.14$. Then we add different quantum noises to $\rho'$ and $\sigma'$ to realize QRDP and we get $\rho''$ and $\sigma''$.

\textbf{Generalized amplitude damping.}
According to Eq.(\ref{gadeps}) and (\ref{fid}), under the GAD mechanism, fidelity and $(\alpha,\epsilon)$-QRDP is related to two parameters, $p$ and $\gamma$. To simplex the mechanism as~\cite{zhou2017differential}, we assume that $p=0.5$. Figure~\ref{fig:fid_gad_gamma} shows that when $\gamma$ increases, the privacy budget and fidelity decrease. Increasing $\gamma$ introduces more noise, which improves privacy protection but diminishes data utility. In addition, it also shows that the privacy budget increases as $\alpha$ increases.

\begin{figure}[t]
    \centering
    \includegraphics[width=0.45\textwidth]{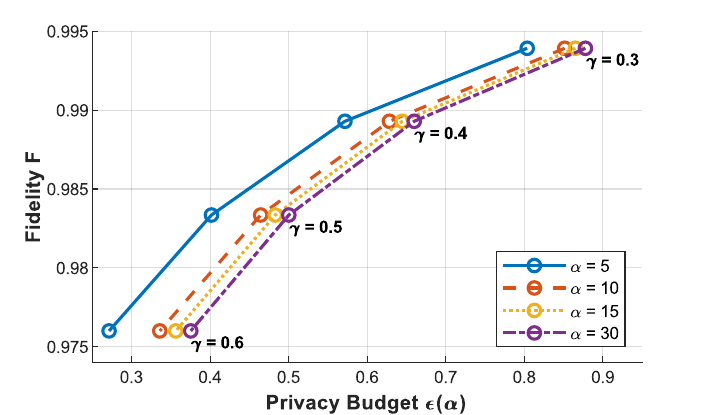}
    \caption{Relationship between $\epsilon$ and $F$ under generalized amplitude damping $\gamma$.}
    \label{fig:fid_gad_gamma}
    \vspace{-2mm}
\end{figure}



\textbf{Composition of phase and amplitude damping.}
According to Eq.(\ref{padeps}) and (\ref{fid}), under the GAD mechanism, fidelity and $(\alpha,\epsilon)$-QRDP is related to three parameters, $p$, $\gamma$ and $\lambda$. To simplex the mechanism as~\cite{zhou2017differential}, we assume that $p=0.5$. Throughout the simulation, the circuit parameters are set to $\gamma = 0.3$ and $\lambda=0.2$ unless otherwise stated. All parameters take values between 0 and 1. Figure~\ref{fig:fid_pad_gamma} and~\ref{fig:fid_pad_lambda} show that when $\gamma$ or $\lambda$ increases, the privacy budget and fidelity decrease. $\gamma$ or $\lambda$ increase indicates an increase in noise and therefore better privacy protection but lower data utility. In addition, it also shows that the privacy budget increases as $\alpha$ increases.

\begin{figure}[t]
    \centering
    \includegraphics[width=0.45\textwidth]{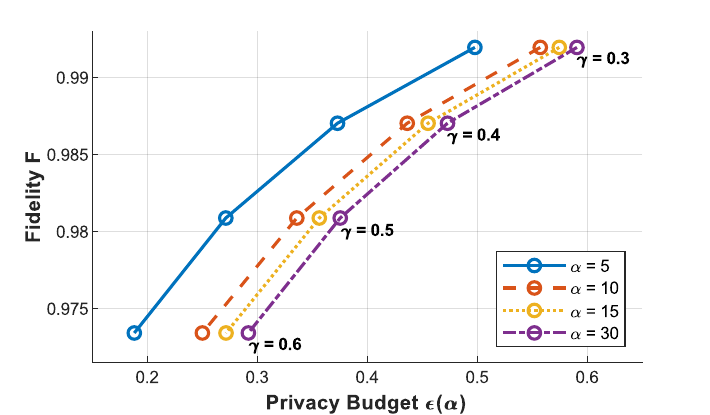}
    \caption{Relationship between $\epsilon$ and $F$ under composition of phase and damping $\gamma$.}
    \label{fig:fid_pad_gamma}
    \vspace{-2mm}
\end{figure}

\begin{figure}[t]
    \centering
    \includegraphics[width=0.45\textwidth]{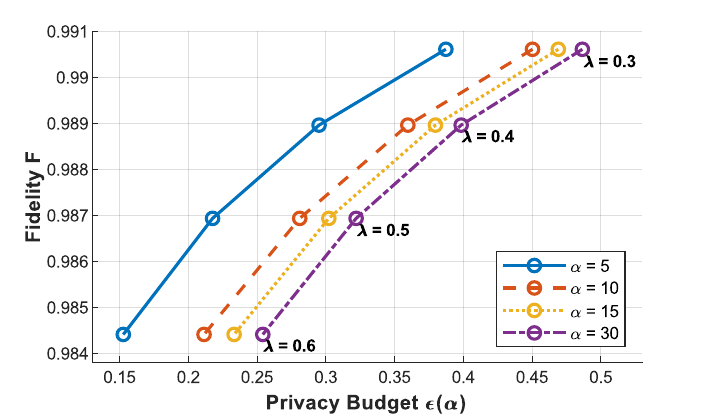}
    \caption{Relationship between $\epsilon$ and $F$ under composition of phase and damping $\lambda$.}
    \label{fig:fid_pad_lambda}
    \vspace{-2mm}
\end{figure}

\textbf{Depolarizing mechanism.}
According to Eq.(\ref{depeps}) and (\ref{fid}), under the Dep mechanism, when the quantum dataset and circuit are determined, $D$ is fixed ($D=2$ for the single qubit input), so $(\alpha,\epsilon)$-QRDP and fidelity are only related to $p$. Figure~\ref{fig:fid_dep_p} shows that when $p$ increases, the privacy budget and fidelity decrease. With an increase in $p$, noise also increases, thus improving privacy protection while decreasing data utility. In addition, it also shows that the privacy budget increases as $\alpha$ increases.

\begin{figure}[t]
    \centering
    \includegraphics[width=0.45\textwidth]{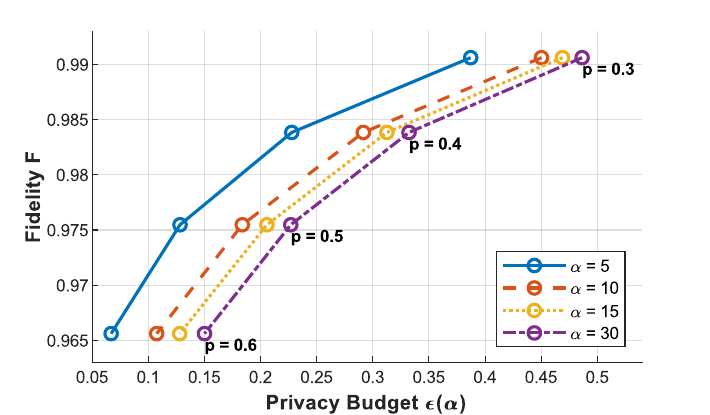}
    \caption{Relationship between $\epsilon$ and $F$ under depolarizing mechanism $p$.}
    \label{fig:fid_dep_p}
    \vspace{-2mm}
\end{figure}

\section{Discussion\label{sec:discussion}}
\subsection{
Rationale for the proposed quantum Rényi divergence
}
In preparation for QRDP, we found that quantum Rényi divergence has already been proposed in~\cite{tomamichel2015quantum} from the perspective of quantum information theory. The formula is
\begin{equation}
    D_{\alpha}(\rho \|\sigma) = \frac{1}{\alpha - 1} \log \text{Tr}(\rho^{\alpha} \sigma^{\alpha-1}),
    \label{oriqrd}
\end{equation}
which is used to measure the difference between two quantum states $\rho$ and $\sigma$, and has been widely used in various fields such as entropy measurement~\cite{goldfeld2024quantum} and quantum cryptography~\cite{pirandola2020advances}. However, when we aim to propose QRDP, we find that Eq.(\ref{oriqrd}) is not suitable for defining QRDP directly. The core of classical DP is to preserve privacy by comparing the difference between neighboring inputs in the distribution of outputs. Quantum states themselves do not have probability distributions in the classical sense, so directly comparing the difference between two quantum states cannot effectively reflect the risk of privacy leakage after quantum operation. For this reason, we propose a new quantum Rényi divergence in Eq.(\ref{qrdivergence}), which focuses on the difference in the distribution of the outcome of quantum states after multiple measurements. Equation (\ref{qrdivergence}) has the additional step of quantum measurements on quantum states over Eq.(\ref{oriqrd}). This approach is similar to the analysis of the output distribution in classical DP. Therefore, Eq.(\ref{qrdivergence}) is more suitable as a basis for the definition of QRDP.

\subsection{Quantum inherent noise to realize QRDP}
A common way to implement QDP or QRDP is to add quantum noise to the model. Existing work has found that the inherent noise during quantum computing can be used as the noise source~\cite{zhou2017differential, guan2023detecting, hirche2023quantum}. This is a promising direction as noise is always considered harmful, but now it provides a certain degree of privacy protection. Since the inherent noise is fixed, the privacy budget is also fixed. Recent research has found some ways to adjust the amount of inherent noise to meet the noise requirements of different QDP in quantum computing~\cite{ju2024quantum,li2024differential,ju2024privacy,ju2024harnessing,zhong2023tuning}.

Quantum error correction (QEC) is used to reduce the noise error rate, and thus reduce the total amount of noise. This approach is suitable for DP calculation related to noise error rate to achieve an adjustable degree of privacy protection, such as the \textit{Theorem 3} from~\cite{zhou2017differential}. Zhong et al.~\cite{zhong2023tuning} calculated the privacy budget $\epsilon$ under the influence of depolarizing noise based on this theorem. They used the Steane code as an example to obtain different noise error rates by varying the position of QEC in the circuit. Different simulation results show that QEC is an effective way to achieve adjustable QDP.

Differing from QEC, quantum error mitigation (QEM) performs post-processing using data from the noisy quantum algorithm. Probabilistic error cancellation (PEC)~\cite{temme2017error}, one of QEM methods, uses physical properties to counteract or minimize the errors of qubits, thereby improving the stability and reliability of quantum systems. A recent work~\cite{ju2024harnessing} used PEC to reduce deviation from the noise and improve the accuracy of the circuit output. They achieved different levels of privacy protection by varying the number of circuit executions or the error rate of depolarizing noise.

Tunable privacy budgets are a new direction for QDP, but the limited number of ways to artificially alter the amount of noise currently limits the scope of privacy budgets. We believe that this can be applied to QRDP as well, since inherent noise is introduced during each QPU operation. Each of these inherent noises leads to a certain degree of privacy preservation, thus allowing QRDP to preserve privacy while maintaining high fidelity. The mechanisms by which quantum inherent noise enables variations in QDP or QRDP remain areas of ongoing exploration.

\section{Conclusion\label{sec:conclusion}}
In this paper, we have novelly developed quantum Rényi differential privacy (QRDP) based on the quantum Rényi divergence. Compared with QDP, we make QRDP adjustable between $\epsilon$-QDP and ($\epsilon$,$\delta$)-QDP by introducing the parameter $\alpha$. QRDP also inherits post-processing and the composition properties, enabling more precise control over the cumulative privacy budget after multiple quantum operations particularly for QDC. We also analyze a variety of noise mechanisms for realizing QRDP, and derive the formulas to calculate the lowest privacy budget under different noise mechanisms. Our numerical simulation results and analysis illustrate the impacts of different quantum noise parameters and the parameter $\alpha$ on the privacy budget. The results also show that when the noise increases, the data protection is stronger while the data utility is weaker. Consequently, the appropriate quantum noise parameter should be selected, in order to achieve the balance between data privacy and data utility.

\bibliographystyle{ACM-Reference-Format}
\bibliography{ref}
\newpage
\appendix
\section{Appendix: Theorem 5}
We provide a more detailed explanation of Theorem~\ref{fidelitytrends}, outlining the fidelity trends under three different noise mechanisms.

\begin{theorem}[\bf{Fidelity Trends under Generalized Amplitude Damping}\label{fidelitytrendsa}]
Let $\rho' = \begin{bmatrix} a & b \\ b^* & c \end{bmatrix}$ be a $2\times 2$ quantum density matrix and let $\rho''= \begin{bmatrix} 
a + \frac{1}{2}(c-a)\gamma & b\sqrt{1-\gamma} \\
b^*\sqrt{1-\gamma} & c + \frac{1}{2}(a-c)\gamma 
\end{bmatrix}$ be the density matrix parameterized by $\gamma$. $a+c=1$ and $b \in \mathbb{C}$. The fidelity $\mathcal{F}(\rho', \rho'')$ between $\rho'$ and $\rho''$ is a decreasing function of $\gamma$, i.e.,
\begin{equation}
    \frac{dF(\rho', \rho'')}{d\gamma} \leq 0 \quad \forall \gamma \in [0, 1].
\end{equation}
\end{theorem}

\begin{theorem}[\bf{Fidelity Trends under Composition of Phase and Amplitude Damping}\label{fidelitytrendsb}]
Let $\rho' = \begin{bmatrix} a & b \\ b^* & c \end{bmatrix}$ be a $2\times 2$ quantum density matrix and let $\rho''= \begin{bmatrix} 
a + \frac{1}{2}(c-a)\gamma & b\sqrt{1-\gamma}\sqrt{1-\lambda} \\
b^*\sqrt{1-\gamma}\sqrt{1-\lambda} & c + \frac{1}{2}(a-c)\gamma 
\end{bmatrix}$ be the density matrix parameterized by $\gamma$ or $\lambda$. $a+c=1$ and $b \in \mathbb{C}$. The fidelity $\mathcal{F}(\rho', \rho'')$ between $\rho'$ and $\rho''$ is a decreasing function of $\gamma$ or $\lambda$, i.e.,
\begin{equation}
    \frac{dF(\rho', \rho'')}{d\gamma} \leq 0 , \quad \frac{dF(\rho', \rho'')}{d\lambda} \leq 0  \quad \forall \gamma,\lambda \in [0, 1].
\end{equation}
\end{theorem}

\begin{theorem}[\bf{Fidelity Trends under Depolarizing Mechanism}\label{fidelitytrendsc}]
Let $\rho' = \begin{bmatrix} a & b \\ b^* & c \end{bmatrix}$ be a $2\times 2$ quantum density matrix and let $\rho''= \begin{bmatrix} 
\frac{pI}{D} + (1 - p)a & b(1-p) \\
b^*(1-p) & \frac{pI}{D} + (1 - p)c 
\end{bmatrix}$ be the density matrix parameterized by $p$. $a+c=1$ and $b \in \mathbb{C}$. The fidelity $\mathcal{F}(\rho', \rho'')$ between $\rho'$ and $\rho''$ is a decreasing function of $p$, i.e.,
\begin{equation}
    \frac{dF(\rho', \rho'')}{dp} \leq 0 \quad \forall p \in [0, 1].
\end{equation}
\end{theorem}

\subsection{Proof of Theorem~\ref{fidelitytrendsa}}
We start by proving Theorem~\ref{fidelitytrendsa}.

\textbf{Step 1: Fidelity Formula for Qubits}
The fidelity of a single qubit system can be expressed by Eq.(\ref{fid}), or the geometric relationship between Bloch vectors.
For two qubit density matrices \( \rho' \) and \( \rho'' \) with Bloch vectors \( \vec{r} \) and \( \vec{s} \), the fidelity \( F(\rho', \rho'') \) is given by~\cite{chen2002bloch}:
\[
F(\rho', \rho'') = \frac{1 + \vec{r} \cdot \vec{s} + \sqrt{(1 - |\vec{r}|^2)(1 - |\vec{s}|^2)}}{2}.
\]

The Bloch vector \( \vec{r} \) for \( \rho' \) is:
\[
\vec{r} = (2b, 0, 2a - 1).
\]

The Bloch vector \( \vec{s} \) for \( \rho'' \) is:

\[
\vec{s} = \left( 2b\sqrt{1 - \gamma}, 0, (2a - 1)(1 - \gamma) \right).
\]

\textbf{Step 2: Express the Fidelity in Terms of \( \gamma \)}

The fidelity \( F(\rho', \rho'') \) depends on \( \gamma \) through the dot product \( \vec{r} \cdot \vec{s} \) and the magnitudes \( |\vec{r}|^2 \) and \( |\vec{s}|^2 \). These terms can be computed as follows.

\begin{enumerate}
    \item \textbf{Dot Product} \( \vec{r} \cdot \vec{s} \):

    \[
    \begin{aligned}
    \vec{r} \cdot \vec{s} &=
    (2b)(2b\sqrt{1 - \gamma}) + (2a - 1)(2a - 1)(1 - \gamma) \\&= 4b^2 \sqrt{1 - \gamma} + (2a - 1)^2(1 - \gamma).
    \end{aligned}
    \]

    \item \textbf{Magnitudes} \( |\vec{r}|^2 \) and \( |\vec{s}|^2 \):

    \[
    |\vec{r}|^2 = 4b^2 + (2a - 1)^2.
    \]
    
    \[
    |\vec{s}|^2 = 4b^2(1 - \gamma) + (2a - 1)^2(1 - \gamma)^2.
    \]
\end{enumerate}

Thus, the fidelity becomes:

\begin{equation*}
\begin{aligned}
F(\gamma) =& \frac{1 + 4b^2\sqrt{1 - \gamma} + (2a - 1)^2(1 - \gamma)}{2} + \\&
\frac{\sqrt{(1 - 4b^2)(1 - (2a - 1)^2)(1 - \gamma)}}{2}.
\end{aligned}
\end{equation*}

Define constants:

\[
k_1 = 4b^2, \quad k_2 = (2a - 1)^2, \quad \delta = 1 - k_1 - k_2.
\]

The fidelity can be rewritten as:

\begin{equation*}
\begin{aligned}
F(\gamma) = \frac{1 + k_1 \sqrt{1 - \gamma} + k_2(1 - \gamma)+\sqrt{\delta \left(1 - k_1(1 - \gamma) - k_2(1 - \gamma)^2 \right)}}{2}.
\end{aligned}
\end{equation*}

\textbf{Step 3: Derivative of Fidelity with Respect to \( \gamma \)}

To prove that \( F(\gamma) \) decreases as \( \gamma \) increases, we compute the derivative of \( F \) with respect to \( \gamma \).

Let \( t = \sqrt{1 - \gamma} \), so \( \gamma = 1 - t^2 \), and express \( F \) in terms of \( t \):

\begin{equation}
F(t) = \frac{1 + k_1t + k_2t^2 + \sqrt{\delta \left( 1 - k_1t^2 - k_2t^4 \right)}}{2}.
\end{equation}

The derivative of \( F(t) \) with respect to \( t \) is:

\begin{equation}
\frac{dF}{dt} = \frac{1}{2} \left( k_1 + 2k_2t - \frac{\delta \left( 2k_1t + 4k_2t^3 \right)}{2 \sqrt{\delta \left( 1 - k_1t^2 - k_2t^4 \right)}} \right).
\end{equation}

The first part of the derivative consists of positive terms \( k_1 \) and \( 2k_2t \), both of which increase linearly with \( t \). The second part of the derivative is a negative term, which involves both linear and cubic terms in \( t \), i.e.,

\[
- \frac{\delta \left( 2k_1t + 4k_2t^3 \right)}{2 \sqrt{\delta \left( 1 - k_1t^2 - k_2t^4 \right)}}.
\]
As \( t \) increases, this negative term grows faster than the positive terms due to the cubic term \( t^3 \).

\textbf{Step 4: Behavior as \( t \) Increases}

As \( t \) increases (i.e., as \( \gamma \) decreases), the negative term dominates the positive terms, because the cubic term \( t^3 \) grows faster than the linear terms \( t \). Therefore, \( \frac{dF}{dt} \) becomes negative for all \( t > 0 \).

Since \( t = \sqrt{1 - \gamma} \), and \( \frac{dF}{d\gamma} = \frac{dF}{dt} \cdot \frac{dt}{d\gamma} \), we know that 

\[
\frac{dt}{d\gamma} = -\frac{1}{2\sqrt{1 - \gamma}},
\]
which is always negative. Hence, 

\[
\frac{dF}{d\gamma} \leq 0.
\]

\textbf{Conclusion:} Since \( \frac{dF}{d\gamma} \leq 0 \) for all \( \gamma \in [0, 1] \), we conclude that \( F(\gamma) \) is a decreasing function of \( \gamma \). Thus, as \( \gamma \) increases, the fidelity \( F(\rho', \rho'') \) decreases, completing the proof.

The proof process of Theorem~\ref{fidelitytrendsb} and~\ref{fidelitytrendsc} is similar to the above. Finally, we can conclude that the fidelity decreases when the noise parameter increases.







\end{document}